\documentclass[11pt,a4paper]{article}

\usepackage[utf8]{inputenc}
\usepackage[margin=1in]{geometry}
\usepackage{amsmath,amssymb,amsfonts}
\usepackage{amsthm}
\usepackage{graphicx}
\usepackage{textcomp}
\usepackage{xcolor}
\usepackage{enumitem}
\setlist[enumerate]{label=\arabic*)}
\usepackage[colorlinks=true,linkcolor=blue,citecolor=blue,urlcolor=blue]{hyperref}

\newtheorem{theorem}{Theorem}
\newtheorem{proposition}[theorem]{Proposition}

\newtheorem{corollary}[theorem]{Corollary}
\theoremstyle{definition}
\newtheorem{definition}[theorem]{Definition}
\newtheorem{example}[theorem]{Example}
\theoremstyle{remark}
\newtheorem{remark}[theorem]{Remark}

\newcommand{\A}{\mathcal{A}}
\newcommand{\C}{\mathcal{C}}
\newcommand{\D}{\mathcal{D}}
\newcommand{\E}{\mathcal{E}}
\newcommand{\R}{\mathcal{R}}
\newcommand{\U}{\mathcal{U}}

\newcommand{\mF}{\mathcal{F}}
\newcommand{\F}{\mathbb{F}}
\newcommand{\Fq}{\mathbb{F}_q}
\newcommand{\Fqm}{\mathbb{F}_{q^m}}
\newcommand{\Fqmu}{\mathbb{F}_{q^\mu}}
\newcommand{\Fqmn}{\mathbb{F}_q^{m\times n}}
\newcommand{\Gab}{\mathbf{Gab}}
\DeclareMathOperator{\rk}{rk}
\DeclareMathOperator{\wt}{wt}
\newcommand{\addremove}{\mathsf{AddRemove}}
\newcommand{\miranda}{\mathsf{MIRANDA}}
\newcommand{\PsiG}{\Psi_{\gamma}}
\newcommand{\drk}{\textnormal{d}_{\mathrm{rk}}}

\DeclareMathOperator{\dd}{d}

\title{Code distances of matrix rank-metric codes under Add-and-Remove transformations}

\author{
Gianira N. Alfarano\\
{\small Universit\'e de Rennes, IRMAR, Rennes, France}\\
{\small\texttt{gianira-nicoletta.alfarano@univ-rennes.fr}}
\and
Rakhi Pratihar\\
{\small School of Mathematical \& Statistical Sciences}\\
{\small Indian Institute of Technology Mandi, Mandi, India}\\
{\small\texttt{rakhi@iitmandi.ac.in}}
\and
Adrien Vin\c{c}otte\\
{\small Universit\'e de Rennes, IRMAR, Rennes, France}\\
{\small\texttt{adrien.vincotte@univ-rennes.fr}}
}

\date{}

\begin{document}

\maketitle

{\renewcommand{\thefootnote}{}\footnotetext{G.N.A. is supported by the Agence Nationale de la Recherche through grant number ANR-24-CPJ1-0075-01. R.P. has been supported by the ANR-DFG through grant number ANR-22-CE92-0075. A.V. acknowledges the support of the DGA and the CreachLab.}}

\begin{abstract}
For a matrix rank-metric code $\C$ and an integer $i$, the $i$-th subcode distance of $\C$ is the largest minimum rank distance of an $i$-dimensional subcode of $\C$. We study how subcode distances change under the Add-and-Remove construction, in which a subcode of codimension $\ell_s$ is kept and a space of dimension $\ell_a$ is added; this construction is used, with Gabidulin codes, in the $\miranda$ signature scheme. For arbitrary matrix codes, we give inequalities between the subcode distances of the original and of the new code, and we show that Delsarte duality interchanges $\ell_s$ and $\ell_a$. For matrix Gabidulin codes, whose subcode distances are given by an explicit formula, we show that, if $\ell_s<\max\{m,n\}$, then within each block of indices $(a-1)\max\{m,n\}+1,\ldots,a\max\{m,n\}$ not exceeding $\dim_{\Fq}(\C_s)$, all but the last $\ell_s$ subcode distances are the same as for the Gabidulin code, and the remaining ones decrease by at most one; a small example shows that such a decrease can occur. We also express the minimum distance of the new code in terms of the rank distance between the added space and the common subcode. For $\miranda$, the lower bound on the minimum distance of the dual public code used by the known distinguisher when $\ell_s=0$ is a special case of these results; for the proposed parameters, where $\ell_s>0$, the same bound holds for a subcode of codimension $\ell_s$ of the dual public code.
\end{abstract}


\section{Introduction}
Rank-metric codes, introduced by Delsarte and Gabidulin~\cite{Delsarte78,Gabidulin85}, are classical objects in coding theory, with applications to random network coding~\cite{SilvaKK08} and to code-based cryptography~\cite{GabidulinPT91,GaboritZ16,Loidreau17}; see also the surveys~\cite{BartzEtAl22,GorlaR18}.  A central challenge in rank-metric cryptography is designing public codes that hide an exploitable algebraic structure while preserving good distance properties. The \emph{Add-and-Remove} construction of~\cite{CouvreurDGV26} produces such a code by keeping a large subcode of a structured code and adding a space that intersects the original code trivially; it is the basis of the $\miranda$ signature scheme~\cite{CouvreurDGV26}. It is then natural to ask how much of the structure of the original code is still visible in the new one.

The minimum distance alone does not answer this question. Finer invariants for rank-metric codes have been studied: generalized and relative generalized rank weights~\cite{Ravagnani16gw,KuriharaMU15}, $q$-polymatroidal invariants~\cite{GorlaJLR20}, and invariants related to rank-metric equivalence~\cite{NeriPH20}. Recently, Camps-Moreno, Gorla, and L\'opez introduced \emph{code distances}, obtained by optimizing minimum distance over subcodes or supercodes of prescribed dimension~\cite{CampsMorenoGL25}. In particular, the subcode distances $\alpha_i$, where $\alpha_i$ is the largest minimum distance of an $i$-dimensional subcode, can distinguish some inequivalent MRD codes having the same parameters~\cite{CampsMorenoGL25}. They are therefore natural candidates for detecting hidden structure.

In this paper we determine which subcode distances are necessarily preserved by the Add-and-Remove construction. Starting from a matrix code $\C$, the construction keeps a subcode $\C_s\leq\C$ of codimension $\ell_s$ and adds a subspace $\A$ of dimension $\ell_a$ with $\A\cap\C=\{0\}$, yielding $\D = \C_s \oplus \A$. In the version of $\miranda$ based on Gabidulin codes, this is used to hide the extension-field structure of a Gabidulin code~\cite{CouvreurDGV26}. We use repeatedly the fact that $\C\cap\D=\C_s$, so that $\ell_s$ and $\ell_a$ are the codimensions of $\C\cap\D$ in $\C$ and in $\D$, respectively.

Our contributions are as follows.

\noindent\textbf{General matrix codes.}
We give inequalities between the subcode distances of $\C$, $\C_s$ and $\D$, and we show that the dual code $\D^\perp$ is obtained from the dual code $\C^\perp$ by the same construction with $\ell_s$ and $\ell_a$ interchanged.

\noindent\textbf{Matrix Gabidulin codes.}
Using the chain of Gabidulin subcodes, we compute all subcode distances of a Gabidulin code. For the Add-and-Remove code, the inequalities above then become equalities for most indices. More precisely, if $\ell_s<\mu=\max\{m,n\}$, then for $i\le\dim_{\Fq}(\C_s)$ the value $\alpha_i(\D)$ can differ from $\alpha_i(\C)$ only if $i$ is one of the last $\ell_s$ indices among $(a-1)\mu+1,\ldots,a\mu$ for some $a$, and in that case it is smaller by at most one. For instance, if $\ell_s=1$, then $\alpha_i(\D)=\alpha_i(\C)$ whenever $\mu\nmid i$. A small example shows that the decrease by one does occur.

\noindent\textbf{Minimum distance and the MRD property.}
The minimum distance $\dd(\D)=\alpha_{\dim\D}(\D)$ also depends on the position of $\A$ with respect to $\C_s$. We express $\dd(\D)$ in terms of the rank distance between $\A$ and $\C_s$ and characterize when $\D$ is MRD. We also consider the case where $\C_s$ is the Gabidulin code of $\Fqmu$-dimension one less; then all subcode distances $\alpha_i(\D)$ with $i\le\dim_{\Fq}(\C_s)$ attain the Singleton bound.

\noindent\textbf{Application to MIRANDA.}
Rank-metric cryptosystems built from Gabidulin codes have a long history~\cite{GabidulinPT91,Loidreau17}, and their algebraic structure has also been the source of structural attacks~\cite{Overbeck08}. In the special case $\ell_s=0$, the lower bound on the minimum distance of the dual public code used in the distinguisher of~\cite{CouvreurDGV26} is a special case of our bounds. For the proposed parameter sets, where $\ell_s>0$, the same bound holds for a subcode of codimension $\ell_s$ of the dual public code, and we obtain similar bounds for the other subcode distances.

Several questions remain open: how often the decrease at the exceptional indices occurs, in particular for the $\miranda$ parameters, what happens for MRD codes that are not Gabidulin codes, and whether subcode distances can be used in a polynomial-time distinguisher. The last question requires knowledge of the subcode distances of random matrix codes, which we leave for future work.

\section{Preliminaries}\label{sec:preliminaries}

Throughout, $q$ is a prime power, all matrix codes are $\Fq$-linear subspaces of $\Fqmn$, and $\mu:=\max\{m,n\}$, $\nu:=\min\{m,n\}$.

\subsection{Matrix rank-metric codes}

For $X,Y\in\Fqmn$, the \textbf{rank distance} between $X,Y$ is defined as $\drk(X,Y):=\rk(X-Y)$,
and the \textbf{rank weight} of $X$ is $\wt_{\mathrm{rk}}(X):=\rk(X)$. A \textbf{matrix rank-metric code} is an $\Fq$-linear subspace $\C\le\Fqmn$. If $\dim_{\Fq}(\C)=k$, we say that $\C$ has parameters $[m\times n,k]_q$. Its minimum rank distance is
$\dd(\C):=\min\{\drk(X,Y):X,Y\in\C,\ X\ne Y\}$. Since $\C$ is linear, $\dd(\C)=\min\{\rk(X):X\in\C\setminus\{0\}\}$.

The parameters of a linear rank-metric code are related by the Singleton-like bound~\cite{Delsarte78,Roth91}, which states that a code $\C\le\Fqmn$ with dimension $k$ and minimum distance $d$ satisfies
\begin{equation}\label{eq:singleton-cardinality}
 k\le
 \max\{m,n\}\bigl(\min\{m,n\}-d+1\bigr)
 ={\mu(\nu-d+1)}.
\end{equation}

Equivalently,
\begin{equation}\label{eq:singleton-distance}
 \dd(\C)\le
 \min\{m,n\}
 -\left\lceil\frac{k}{\max\{m,n\}}\right\rceil+1
 =\nu-\left\lceil\frac{k}{\mu}\right\rceil+1.
\end{equation}
A linear code is called a \textbf{maximum rank distance (MRD) code} if equality holds in \eqref{eq:singleton-cardinality}.
In particular, the dimension of an $\Fq$-linear MRD code is a multiple of $\mu=\max\{m,n\}$.

We use the standard trace inner product $\langle X,Y\rangle:=\operatorname{Tr}(XY^{\mathsf T})$ on $\Fqmn$. The \textbf{(Delsarte) dual} of $\C$ is $\C^\perp:=\{Y\in\Fqmn:\langle X,Y\rangle=0\text{ for every }X\in\C\}$, and $\dim_{\Fq}(\C^\perp)=mn-\dim_{\Fq}(\C)$.

\subsection{Vector rank metric, matrix expansion, and Gabidulin codes}

Vector rank-metric codes over $\Fqmu$ are related to matrix codes via the standard expansion~\cite{Gabidulin85}. For $x=(x_1,\ldots,x_\nu)\in\Fqmu^\nu$, set $\wt_{\mathrm{rk}}(x):=\dim_{\Fq}\langle x_1,\ldots,x_\nu\rangle_{\Fq}$ and $\drk(x,y):=\wt_{\mathrm{rk}}(x-y)$. Fix an ordered $\Fq$-basis $\gamma=(\gamma_1,\ldots,\gamma_\mu)$ of $\Fqmu$. If $x_j=\sum_{r=1}^{\mu}x_{rj}\gamma_r$ with $x_{rj}\in\Fq$, set $M_\gamma(x):=(x_{rj})_{r,j}\in\Fq^{\mu\times\nu}$. To obtain an $m\times n$ matrix, define the \textbf{$\gamma$-expansion map}
\begin{equation}\label{eq:expansion}
 \PsiG:\Fqmu^\nu\longrightarrow\Fqmn,
 \qquad
 \PsiG(x):=
 \begin{cases}
  M_\gamma(x),& m\ge n,\\
  M_\gamma(x)^{\mathsf T},& m<n.
 \end{cases}
\end{equation}
The map \eqref{eq:expansion} is an $\Fq$-linear rank-metric isometry: $\rk(\PsiG(x))=\wt_{\mathrm{rk}}(x)$. Thus an $\Fqmu$-linear code of dimension $\kappa$ expands to an $\Fq$-linear matrix code of dimension $\mu\kappa$ with the same minimum rank distance.

Let $m\geq n$ (the case $m<n$ follows by transposition). A $q$-\textbf{linearized polynomial} over $\Fqm$ has the form $f(x)=\sum_{j=0}^{r} f_j x^{q^j}$ with $f_j\in\Fqm$. If $f\ne0$, its $q$-degree is $\deg_q(f):=\max\{j:f_j\ne0\}$. Such a polynomial induces an $\Fq$-linear map on $\Fqm$.

\begin{definition}
Let $g=(g_1,\ldots,g_n)\in\Fqm^n$ have $\Fq$-linearly independent entries, and let $1\le\kappa\le n$. The \textbf{Gabidulin code} of dimension $\kappa$ evaluated at $g$ is $\Gab(g,\kappa):=\{(f(g_1),\ldots,f(g_n)): \deg_q(f)<\kappa\}$.
\end{definition}

The code $\Gab(g,\kappa)$ has $\Fqm$-dimension $\kappa$ and minimum rank distance $n-\kappa+1$~\cite{Gabidulin85}; hence $\C:=\PsiG(\Gab(g,\kappa))$ has parameters $[m\times n,m\kappa]_q$ and is MRD.

We will use the following chain of Gabidulin subcodes. For $1\le a\le\kappa$, set $\C_a:=\PsiG(\Gab(g,a))$. Then $\C_1\subseteq\C_2\subseteq\cdots\subseteq\C_\kappa=\C$, $\dim_{\Fq}(\C_a)=m a$, and
\begin{equation}\label{eq:gab-subcode}
 \dd(\C_a)=n-a+1.
\end{equation}

\subsection{Add-and-Remove construction}

\begin{definition}[Add-and-Remove construction \cite{CouvreurDGV26}]
Let $\C\le\Fqmn$ be a matrix code of dimension $k$. Let $\C_s\le\C$ have codimension $\ell_s$, and let $\A\le\Fqmn$ have dimension $\ell_a$ and satisfy $\A\cap\C=\{0\}$. The corresponding \textbf{Add-and-Remove code} is $\D:=\C_s\oplus\A$.
It has dimension $\dim_{\Fq}(\D)=k-\ell_s+\ell_a$.
We write $\D\in\addremove(\C,\ell_a,\ell_s)$.
\end{definition}

In $\miranda$, $\C_s$ need only be an $\Fq$-subcode of the matrix Gabidulin code~\cite{CouvreurDGV26}. Moreover, $\C\cap\D=\C_s$: if $X=S+A\in\C\cap\D$ with $S\in\C_s$ and $A\in\A$, then $A=X-S\in\A\cap\C$, so $A=0$. Thus $\ell_s$ and $\ell_a$ are the codimensions of the common subcode in $\C$ and $\D$, respectively. Conversely, for any two matrix codes $\C,\D\le\Fqmn$, taking $\C_s:=\C\cap\D$ and any complement $\A$ of $\C_s$ in $\D$ gives $\A\cap\C=\A\cap\C_s=\{0\}$, so $\D\in\addremove(\C,\ell_a,\ell_s)$ with $\ell_s=\dim_{\Fq}(\C)-\dim_{\Fq}(\C\cap\D)$ and $\ell_a=\dim_{\Fq}(\D)-\dim_{\Fq}(\C\cap\D)$. Hence $\C_s=\C\cap\D$, $\ell_s$ and $\ell_a$ depend only on the pair $(\C,\D)$, not on the choice of~$\A$.

For later use, for $X\in\Fqmn$ and a matrix code $\U\le\Fqmn$, define $\drk(X,\U):=\min_{U\in\U}\rk(X-U)$. If $\A\le\Fqmn$ is nonzero and $\A\cap\U=\{0\}$, set $\Delta(\A,\U):=\min_{0\ne X\in\A}\drk(X,\U)$, and use the convention $\Delta(\{0\},\U)=+\infty$. Equivalently, $\Delta(\A,\U)$ is the minimum rank of an element of $(\A\oplus\U)\setminus\U$; in particular, it depends on $\A$ only through $\A\oplus\U$.

\subsection{Subcode distances}

We now introduce the invariant studied in the paper.

\begin{definition}[Subcode distances \cite{CampsMorenoGL25}]
Let $\C\le\Fqmn$ have dimension $k$. For $1\le i\le k$, the $i$-th \textbf{subcode distance} of $\C$ is $\alpha_i(\C):=\max\{\dd(\E):\E\le\C,\ \dim_{\Fq}(\E)=i\}$.
\end{definition}

Clearly $\alpha_k(\C)=\dd(\C)$ and $\alpha_1(\C)\ge\cdots\ge\alpha_k(\C)$, since passing to a subcode cannot decrease the minimum distance.
Applying the Singleton bound \eqref{eq:singleton-distance} to each $i$-dimensional subcode gives, for $1\le i\le k$,
\begin{equation}\label{eq:alpha-sing}
 \alpha_i(\C)\le \nu-\left\lceil\frac{i}{\mu}\right\rceil+1.
\end{equation}
For matrix Gabidulin codes, equality holds for every $i$. We include the short proof, since this formula is used throughout the paper.

\begin{proposition}\label{prop:gab-alpha}
Let $m\ge n$ and $\C=\PsiG(\Gab(g,\kappa))\le\Fqmn$. Then, for every $1\le i\le m\kappa$,
\begin{equation}\label{eq:gab-alpha}
 \alpha_i(\C)=n-\left\lceil\frac{i}{m}\right\rceil+1.
\end{equation}
\end{proposition}

\begin{proof}
Fix $1\le i\le m\kappa$ and set $a:=\lceil i/m\rceil\le\kappa$. Then $(a-1)m<i\le am$. By \eqref{eq:gab-subcode}, the Gabidulin subcode $\C_a\le\C$ has dimension $am$ and distance $n-a+1$. Any $i$-dimensional subcode $\E\le\C_a$ satisfies $\dd(\E)\ge \dd(\C_a)=n-a+1$. Hence $\alpha_i(\C)\ge n-a+1$. The opposite inequality is~\eqref{eq:alpha-sing}.
\end{proof}

\section{Subcode Distances under Add-and-Remove}\label{sec:general}

We first give inequalities that hold for arbitrary matrix codes. Their proofs only use the inclusions $\C_s\le\C$ and $\C_s\le\D$ and the codimensions $\ell_s$ and $\ell_a$.

\begin{proposition}\label{prop:general}
Let $\C$ be a matrix code of dimension $k$ and let $\D=\C_s\oplus\A\in\addremove(\C,\ell_a,\ell_s)$. Then
\begin{align}
 \alpha_i(\C_s)&\ge \alpha_{i+\ell_s}(\C),
 &&1\le i\le k-\ell_s,\label{eq:remove-bound}\\
 \alpha_i(\D)&\ge \alpha_i(\C_s),
 &&1\le i\le k-\ell_s,\label{eq:add-lower}\\
 \alpha_i(\D)&\le \alpha_{i-\ell_a}(\C_s),
 &&\ell_a<i\le k-\ell_s+\ell_a.\label{eq:add-upper}
\end{align}
Consequently,
\begin{align}
 \alpha_i(\D)&\ge \alpha_{i+\ell_s}(\C),
 &&1\le i\le k-\ell_s,\label{eq:general-lower}\\
 \alpha_i(\D)&\le \alpha_{i-\ell_a}(\C),
 &&\ell_a<i\le k-\ell_s+\ell_a.\label{eq:general-upper}
\end{align}
\end{proposition}

\begin{proof}
For \eqref{eq:remove-bound}, let $\E\le\C$ have dimension $i+\ell_s$ and satisfy $\dd(\E)=\alpha_{i+\ell_s}(\C)$. Since $\C_s$ has codimension $\ell_s$ in $\C$, the intersection $\E\cap\C_s$ has dimension at least $i$. Choosing an $i$-dimensional subcode $\mF\le\E\cap\C_s$ gives $\alpha_i(\C_s)\ge \dd(\mF)\ge \dd(\E)=\alpha_{i+\ell_s}(\C)$, proving \eqref{eq:remove-bound}. The inclusion $\C_s\le\D$ immediately gives \eqref{eq:add-lower}.

For \eqref{eq:add-upper}, choose an $i$-dimensional subcode $\E\le\D$ with $\dd(\E)=\alpha_i(\D)$. Since $\C_s$ has codimension $\ell_a$ in $\D$, the intersection $\E\cap\C_s$ has dimension $r\ge i-\ell_a>0$. Hence $\alpha_r(\C_s)\ge \dd(\E\cap\C_s)\ge \dd(\E)=\alpha_i(\D)$. Monotonicity gives $\alpha_{i-\ell_a}(\C_s)\ge\alpha_r(\C_s)$, proving \eqref{eq:add-upper}. The bounds \eqref{eq:general-lower} and \eqref{eq:general-upper} then follow from \eqref{eq:remove-bound}--\eqref{eq:add-upper} and the inclusion $\C_s\le\C$.
\end{proof}

The construction behaves well with respect to duality.

\begin{remark}[Duality]\label{prop:dual-addremove}
If $\D\in\addremove(\C,\ell_a,\ell_s)$, then $\D^\perp\in\addremove(\C^\perp,\ell_s,\ell_a)$. Indeed, $\C^\perp\cap\D^\perp=(\C+\D)^\perp$, and $\dim_{\Fq}(\C+\D)=\dim_{\Fq}(\C)+\ell_a$ because $\C\cap\D=\C_s$. Hence $\C^\perp\cap\D^\perp$ has codimension $\ell_a$ in $\C^\perp$ and codimension $\ell_s$ in $\D^\perp$. Moreover, $\D^\perp=\E\oplus\R$, where $\E:=\C^\perp\cap\A^\perp=\C^\perp\cap\D^\perp$ (as $\C^\perp\le\C_s^\perp$) and $\R$ is any complement of $\E$ in $\D^\perp$; then $\dim_{\Fq}(\R)=\ell_s$ and $\R\cap\C^\perp=\R\cap\E=\{0\}$.
\end{remark}

\begin{corollary}\label{cor:comparison}
Let $\D=\C_s\oplus\A\in\addremove(\C,\ell_a,\ell_s)$ and let $\ell_a < i
\le k-\ell_s$. If the subcode distances of $\C$ with indices in $[i-\ell_a,\, i+\ell_s]$ are all equal, that is,
\[
   \alpha_{i-\ell_a}(\C) = \alpha_{i-\ell_a+1}(\C)
   = \cdots = \alpha_{i+\ell_s}(\C) =: v,
\]
then $\alpha_i(\D) = v$. By Remark~\ref{prop:dual-addremove}, the same statement holds for $\D^\perp$ and $\C^\perp$ with $\ell_s$ and $\ell_a$ interchanged.
\end{corollary}

\begin{proof}
Inequalities \eqref{eq:general-lower} and \eqref{eq:general-upper} give
$v = \alpha_{i+\ell_s}(\C) \le \alpha_i(\D) \le \alpha_{i-\ell_a}(\C) = v$.
\end{proof}

\section{Add-and-Remove Gabidulin Codes}\label{sec:gab-addremove}

We now apply Proposition~\ref{prop:general} to Gabidulin codes, using Proposition~\ref{prop:gab-alpha}. Let $1\le\kappa\le n\le m$ and let $\C=\PsiG(\Gab(g,\kappa))$ be a matrix Gabidulin code of $\Fq$-dimension $m\kappa$. Let $\C_s\le\C$ have codimension $\ell_s$, let $\A\le\Fqmn$ have dimension $\ell_a$ with $\A\cap\C=\{0\}$, and set $\D=\C_s\oplus\A$. If $\ell_a>0$, then necessarily $\kappa<n$, since $\kappa=n$ would imply $\C=\Fqmn$.

\subsection{General bounds}
For $i\le\dim_{\Fq}(\C_s)$, we obtain upper and lower bounds on $\alpha_i(\D)$ that differ by at most one when $\ell_s\le m$.

\begin{theorem}\label{thm:gab-general}
For every $1\le i\le m\kappa-\ell_s$, we have
\begin{equation}\label{eq:gab-bounds}
 n-\left\lceil\frac{i+\ell_s}{m}\right\rceil+1
 \le \alpha_i(\D)\le
 n-\left\lceil\frac{i}{m}\right\rceil+1.
\end{equation}
Consequently, if
\begin{equation}\label{eq:same-ceil}
 \left\lceil\frac{i+\ell_s}{m}\right\rceil
 =\left\lceil\frac{i}{m}\right\rceil,
\end{equation}
then $\alpha_i(\D)=n-\left\lceil i/m\right\rceil+1$.
\end{theorem}

\begin{proof}
By \eqref{eq:general-lower} and \eqref{eq:gab-alpha}, $\alpha_i(\D)\ge\alpha_{i+\ell_s}(\C)=n-\left\lceil(i+\ell_s)/m\right\rceil+1$. The upper bound is the Singleton estimate \eqref{eq:alpha-sing} applied to the $i$-dimensional subcodes of $\D$. If \eqref{eq:same-ceil} holds, the two bounds coincide.
\end{proof}

Here \eqref{eq:general-upper} does not improve the upper bound in \eqref{eq:gab-bounds}, since the $\alpha_i(\C)$ attain the Singleton bound. Theorem~\ref{thm:gab-general} does not cover the indices $i>\dim_{\Fq}(\C_s)$, where $\alpha_i(\D)$ can be much smaller (Proposition~\ref{prop:min}).

Suppose now that $0\le\ell_s<m$, and write $i=(a-1)m+t$ with $1\le t\le m$. If $t\le m-\ell_s$, then $i+\ell_s\le am$, so \eqref{eq:same-ceil} holds. Hence, among the indices $(a-1)m+1,\ldots,am$, only the last $\ell_s$ ones can satisfy $\alpha_i(\D)\ne\alpha_i(\C)$, and the two bounds in \eqref{eq:gab-bounds} differ by at most one.

\begin{corollary}\label{cor:blocks}
Assume $0\le\ell_s<m$, and let $i=(a-1)m+t\le m\kappa-\ell_s$ with $1\le t\le m$. Then
\begin{enumerate}
 \item if $1\le t\le m-\ell_s$, then $ \alpha_i(\D)=n-a+1$;
 \item if $m-\ell_s<t\le m$, then
 $\alpha_i(\D)\in\{n-a+1,n-a\}$.
\end{enumerate}
In particular, if $\ell_s=0$, then $\alpha_i(\D)=n-\lceil i/m\rceil+1$ for all $1\le i\le m\kappa$.
\end{corollary}

Both values in case 2) occur; see Example~\ref{ex:f8}.
Since $(\alpha_i(\D))_i$ is nonincreasing, Corollary~\ref{cor:blocks} also shows
where a decrease can occur. For $\C$, such a decrease occurs exactly
when $j$ is a multiple of $m$. For $\D$, if $\alpha_j(\D)=n-a+1$ and $\alpha_{j+1}(\D)=n-a$ with $j+1\le m\kappa-\ell_s$, then $am-\ell_s\le j\le am$. For $\ell_s=1$, this means $j\in\{am-1,am\}$.

For $i=am<m\kappa$, we have $\alpha_{am}(\D)=n-a+1$ if and only if $\D$ contains an MRD subcode of dimension $am$, since the Singleton bound for dimension $am$ is $n-a+1$.

\subsection{The minimum distance}

The added space $\A$ may contain matrices of small rank, or produce them together with $\C_s$; the minimum distance of $\D$ therefore depends on $\A$.

\begin{proposition}\label{prop:min}
Assume $0\le\ell_s<m$. Then
\begin{equation}\label{eq:min}
 \dd(\D)=\min\{n-\kappa+1,\Delta(\A,\C_s)\}.
\end{equation}
If $\ell_a=\ell_s$, then $\D$ has dimension $m\kappa$ and is MRD if and only if
\begin{equation}\label{eq:mrd}
 \Delta(\A,\C_s)\ge n-\kappa+1.
\end{equation}
\end{proposition}

\begin{proof}
Every nonzero element of $\D=\C_s\oplus\A$ belongs either to $\C_s\setminus\{0\}$ or to a coset $X+\C_s$ with $0\ne X\in\A$. Consequently,
\begin{equation}\label{eq:min-decomposition}
 \dd(\D)=\min\{\dd(\C_s),\Delta(\A,\C_s)\}.
\end{equation}
Since $\C_s\le\C$ and $\dd(\C)=n-\kappa+1$, one has $\dd(\C_s)\ge n-\kappa+1$. On the other hand, $\dim_{\Fq}(\C_s)=m\kappa-\ell_s$ and $0\le\ell_s<m$, so $\lceil(m\kappa-\ell_s)/m\rceil=\kappa$. The Singleton bound gives the reverse inequality, hence $\dd(\C_s)=n-\kappa+1$, and \eqref{eq:min} follows from \eqref{eq:min-decomposition}. If $\ell_a=\ell_s$, then $\dim_{\Fq}(\D)=m\kappa$. Such a code is MRD exactly when $\dd(\D)=n-\kappa+1$, which by \eqref{eq:min} is equivalent to \eqref{eq:mrd}.
\end{proof}

Consequently, an Add-and-Remove operation need not preserve the MRD property; see Example~\ref{ex:f8}.

\subsection{Removing one \texorpdfstring{$\Fqm$}{F\_\{q\textasciicircum m\}}-dimension}

We next remove one whole $\Fqm$-dimension, choosing the remaining subcode to be Gabidulin.

\begin{proposition}\label{prop:fqm-removal}
Assume $2\le\kappa<n$, let $\C_0=\PsiG(\Gab(g,\kappa-1))\le\C$, and let $\A\le\Fqmn$ have dimension $m$ with $\A\cap\C=\{0\}$. Set $\D=\C_0\oplus\A$. Then, for every $1\le i\le m(\kappa-1)$,
\begin{equation}\label{eq:fqm-case}
 \alpha_i(\D)=n-\left\lceil\frac{i}{m}\right\rceil+1.
\end{equation}
Moreover,
\begin{equation}\label{eq:fqm-d}
 \dd(\D)=\Delta(\A,\C_0)\le n-\kappa+1.
\end{equation}
In particular, $\D$ is MRD if and only if $\Delta(\A,\C_0)=n-\kappa+1$.
\end{proposition}

\begin{proof}
For $1\le i\le m(\kappa-1)$, the code $\C_0$ contains the nested Gabidulin subcodes $\C_a$ for $a\le\kappa-1$. The proof of Proposition~\ref{prop:gab-alpha} therefore gives a subcode of $\C_0\le\D$ attaining the bound \eqref{eq:alpha-sing} at every such index, proving \eqref{eq:fqm-case}.
The code $\C_0$ has minimum distance $n-\kappa+2$, while $\D$ has dimension $m\kappa$. As in \eqref{eq:min-decomposition}, $\dd(\D)=\min\{n-\kappa+2,\Delta(\A,\C_0)\}$. The Singleton bound for $\D$ gives $\dd(\D)\le n-\kappa+1$. Therefore $\Delta(\A,\C_0)\le n-\kappa+1$, so the minimum is attained by $\Delta(\A,\C_0)$, proving \eqref{eq:fqm-d}. The MRD characterization follows immediately.
\end{proof}

Thus, when $\C_s$ is itself a Gabidulin code, all subcode distances $\alpha_i(\D)$ with $i\le\dim_{\Fq}(\C_s)$ attain the Singleton bound, whereas for other subcodes $\C_s$ of small codimension $\ell_s<m$ some of them can be smaller by one (Corollary~\ref{cor:blocks} and Example~\ref{ex:f8}). In all cases, the minimum distance of $\D$ depends on $\A$.

\section{Application to MIRANDA and an Example}\label{sec:crypto}

We consider the version of $\miranda$ based on Gabidulin codes~\cite{CouvreurDGV26}. There, $q=2$, $m=n$, and the secret key contains a matrix Gabidulin code $\C=\PsiG(\Gab(g,\kappa))$ of $\Fq$-dimension $\kappa m$. One chooses a subcode $\C_s\le\C$ of codimension $\ell_s$ and a code $\A$ of dimension $\ell_a$ with $\A\cap\C=\{0\}$, and sets $\D_{\mathrm{mat}}:=\C_s\oplus\A$. The public key is a basis of $\D_{\mathrm{mat}}^\perp$, which determines $\D_{\mathrm{mat}}$. For all signature parameters of~\cite[Table~1]{CouvreurDGV26}, one has $2\le\ell_s\le5$ and $m\ge79$, so $\ell_s<m$ and the results of Section~\ref{sec:gab-addremove} apply to $\D_{\mathrm{mat}}$. We will illustrate the results with the parameter set $(m,n,\kappa,\ell_a,\ell_s)=(113,113,107,220,3)$ of~\cite[Table~1]{CouvreurDGV26}, which targets $128$-bit security (NIST security level~I); then $\dim_{\Fq}(\C_s)=12088$ and $\dim_{\Fq}(\D_{\mathrm{mat}})=12308$.

We first consider the subcode distances of $\D_{\mathrm{mat}}$. Write $i=(a-1)m+t$ with $1\le t\le m$ and $i\le\kappa m-\ell_s$. By Corollary~\ref{cor:blocks},
\begin{align*}
 \alpha_i(\D_{\mathrm{mat}})&=n-a+1, && t\le m-\ell_s,\\
 \alpha_i(\D_{\mathrm{mat}})&\in\{n-a+1,n-a\}, && t>m-\ell_s.
\end{align*}
Thus, among any $m$ consecutive indices, at most $\ell_s\le5$ subcode distances of $\D_{\mathrm{mat}}$ are not determined by the parameters. For the parameter set above, $\alpha_i(\D_{\mathrm{mat}})=114-\lceil i/113\rceil$ for all $i\le12088$ with $t\le110$. In addition, for every $1\le a\le\kappa$, the code $\C_a\cap\C_s$ is a subcode of $\D_{\mathrm{mat}}$ of dimension between $am-\ell_s$ and $am$ and minimum distance at least $n-a+1$; since $\ell_s<m$, this is the Singleton bound for that dimension.

We now turn to the dual public code. By Remark~\ref{prop:dual-addremove}, $\D_{\mathrm{mat}}^\perp=\E\oplus\R$, where $\E:=\C^\perp\cap\D_{\mathrm{mat}}^\perp$ is a subcode of codimension $\ell_a$ of $\C^\perp$ and $\dim_{\Fq}(\R)=\ell_s$; this is the decomposition used in~\cite[Sec.~5.3]{CouvreurDGV26}.
The code $\C^\perp$ is again an $\Fqm$-linear matrix Gabidulin code, of $\Fqm$-dimension $n-\kappa$ and minimum distance $\kappa+1$: indeed, $\C^\perp=\Psi_{\gamma^*}(\Gab(h,n-\kappa))$ for some $h\in\Fqm^n$, where $\gamma^*$ is the trace-dual basis of $\gamma$, see e.g.~\cite{Gabidulin85,Ravagnani16dual}. Note that $\C^\perp$ is $\Fqm$-linear with respect to $\gamma^*$, but in general not with respect to $\gamma$. Since Proposition~\ref{prop:gab-alpha}, and hence Theorem~\ref{thm:gab-general}, holds for any basis, Theorem~\ref{thm:gab-general} applies to $\D_{\mathrm{mat}}^\perp\in\addremove(\C^\perp,\ell_s,\ell_a)$ and gives, for $1\le j\le m(n-\kappa)-\ell_a$,
\begin{equation}\label{eq:miranda-dual-profile}
 n-\left\lceil\frac{j+\ell_a}{m}\right\rceil+1
 \le \alpha_j(\D_{\mathrm{mat}}^\perp)
 \le n-\left\lceil\frac{j}{m}\right\rceil+1.
\end{equation}
For the largest index $j=\dim_{\Fq}(\E)=m(n-\kappa)-\ell_a$, the lower bound reads
\begin{equation}\label{eq:miranda-dual-gap}
 \alpha_{m(n-\kappa)-\ell_a}(\D_{\mathrm{mat}}^\perp)\ge\kappa+1 ,
\end{equation}
which also follows directly from $\dd(\E)\ge\dd(\C^\perp)=\kappa+1$.
If $\ell_s=0$, then $\D_{\mathrm{mat}}^\perp=\E$ and \eqref{eq:miranda-dual-gap} says that $\dd(\D_{\mathrm{mat}}^\perp)\ge\kappa+1$, i.e., $\D_{\mathrm{mat}}^\perp$ contains no nonzero matrix of rank at most $\kappa$, whereas a random code of the same dimension does with high probability; see~\eqref{eq:chebyshev} below. This is the property used by the exponential-time distinguisher of~\cite[Sec.~5.2]{CouvreurDGV26}, which is described there for $\ell_s=0$. If $\ell_s>0$, then \eqref{eq:miranda-dual-gap} says that $\D_{\mathrm{mat}}^\perp$ contains a subcode of codimension $\ell_s$, namely $\E$, with no nonzero matrix of rank at most $\kappa$.

Unlike the bounds for $\D_{\mathrm{mat}}$, \eqref{eq:miranda-dual-profile} determines no value $\alpha_j(\D_{\mathrm{mat}}^\perp)$ when $\ell_a\ge m$, since then its two sides always differ. This is the case for all sets of~\cite[Table~1]{CouvreurDGV26} except the one with $m=233$; for instance, $\ell_a=220>113$ for the parameter set above.
We now compare these values with those of other codes. A first comparison, valid for every code of a given dimension, follows from Theorem~\ref{thm:gab-general} itself.

\begin{proposition}\label{prop:any-code}
Let $m\ge n$ and let $\U\le\Fqmn$ be any code of codimension $c$. Then, for $1\le i\le mn-c$,
\[
 n-\left\lceil\frac{i+c}{m}\right\rceil+1\le\alpha_i(\U)\le n-\left\lceil\frac{i}{m}\right\rceil+1.
\]
\end{proposition}

\begin{proof}
For every $g\in\Fqm^n$ with $\Fq$-linearly independent entries, $\PsiG(\Gab(g,n))=\Fqmn$, and $\U\in\addremove(\Fqmn,0,c)$. Apply Theorem~\ref{thm:gab-general} with $\kappa=n$, $\ell_s=c$ and $\ell_a=0$.
\end{proof}

A code of the same dimension as $\D_{\mathrm{mat}}$ has codimension $c=m(n-\kappa)-\ell_a+\ell_s$. Hence every such code contains, for each $a$, a subcode of dimension at least $am-c$ and minimum distance at least $n-a+1$ (its intersection with an MRD code of dimension $am$), whereas $\D_{\mathrm{mat}}$ contains such a subcode of dimension at least $am-\ell_s$. The difference $c-\ell_s=m(n-\kappa)-\ell_a$ is exactly $\dim_{\Fq}(\E)$. For the parameter set above, $c=461$, so every code $\U$ of dimension $12308$ satisfies $\alpha_i(\U)\ge114-\lceil(i+461)/113\rceil$, whereas $\alpha_i(\D_{\mathrm{mat}})\ge114-\lceil (i+3)/113\rceil$ for $i\le12088$. Thus the subcode distances of $\D_{\mathrm{mat}}$ can distinguish it from a random code only if, for some $i\le12088$, a random code $\U$ of this dimension satisfies $\alpha_i(\U)<114-\lceil (i+3)/113\rceil$ with high probability, that is, $\alpha_i(\U)$ is smaller than the lower bound on $\alpha_i(\D_{\mathrm{mat}})$; we do not know whether this is the case.

For the dual code, Proposition~\ref{prop:any-code} gives no information, since $\dim_{\Fq}(\D_{\mathrm{mat}}^\perp)$ is small. But, we can use a second-moment argument. Let $\U$ be a uniformly random $K$-dimensional code, let $B_r$ be the set of nonzero matrices of rank at most $r$ in $\Fqmn$, and let $X:=|\U\cap B_r|$. Then $\mathbb{E}[X]=|B_r|(q^K-1)/(q^{mn}-1)$. Two linearly independent matrices lie in $\U$ with probability at most the product of the individual probabilities, and $B_r$ is closed under multiplication by $\Fq^*$; hence $\operatorname{Var}(X)\le(q-1)\,\mathbb{E}[X]$, and Chebyshev's inequality gives
\begin{equation}\label{eq:chebyshev}
 \Pr[X=0]\le\frac{\operatorname{Var}(X)}{\mathbb{E}[X]^2}\le\frac{q-1}{\mathbb{E}[X]}.
\end{equation}

For the parameter set above, $r=\kappa=107$ and $K=\dim_{\Fq}(\D_{\mathrm{mat}}^\perp)=461$ give $\mathbb{E}[X]\approx2^{426.7}$, so a random code of this dimension contains a nonzero matrix of rank at most $\kappa$, except with probability at most $2^{-426}$. If $\ell_s=0$, then $\D_{\mathrm{mat}}^\perp$ contains no such matrix by \eqref{eq:miranda-dual-gap}. If $\ell_s>0$, only the subcode $\E$ of codimension $\ell_s$ is guaranteed to contain none. Whether a random code of dimension $K$ has a subcode of codimension $\ell_s$ without nonzero matrices of rank at most $\kappa$ is not decided by this argument: a union bound over its subcodes of codimension $\ell_s$ multiplies the above probability by $\genfrac{[}{]}{0pt}{2}{K}{\ell_s}_q\approx q^{\ell_s(K-\ell_s)}$, which is too large here. We leave this question open.

These results do not give a new attack: a distinguisher would also need the distribution of subcode distances of random matrix codes and an efficient way to compute them. Nor do they give the left stabilizer algebra of the hidden Gabidulin code, which the key-recovery attack of~\cite[Sec.~5.3]{CouvreurDGV26} recovers before applying classical attacks~\cite{Overbeck08}.

\begin{example}\label{ex:f8}
Let $\F_8=\F_2[\theta]/(\theta^3+\theta+1)$, take $\gamma=(1,\theta,\theta^2)$ and $g=(1,\theta,\theta^2)$, and set $\C:=\PsiG(\Gab(g,2))$. Then $\C\le\F_2^{3\times3}$ has dimension $6$ and minimum distance $2$, and Proposition~\ref{prop:gab-alpha} gives $(\alpha_1(\C),\ldots,\alpha_6(\C))=(3,3,3,2,2,2)$.

For a linearized polynomial $f$, write $\operatorname{ev}_g(f):=(f(1),f(\theta),f(\theta^2))$. Let $\mathcal L_0:=\{a_0x+a_1x^2:\ a_0\in\F_8,\ a_1\in\langle1,\theta\rangle_{\F_2}\}$ and $\C_0:=\PsiG(\operatorname{ev}_g(\mathcal L_0))$. Then $\C_0$ is a $5$-dimensional $\F_2$-hyperplane of $\C$ and contains the $3$-dimensional matrix Gabidulin subcode $\PsiG(\Gab(g,1))$, whose minimum distance is $3$. Let $A=E_{11}$, whose inverse image under $\PsiG$ is $(1,0,0)$. If $A\in\C$, some $f(x)=a_0x+a_1x^2$ with $a_0,a_1\in\F_8$ would satisfy $f(\theta)=f(\theta^2)=0$ and $f(1)=1$. This is impossible, because such an $f$ is nonzero, and a nonzero linearized polynomial of $q$-degree at most one (here $q=2$) has $\F_2$-kernel of dimension at most one, while $\theta,\theta^2$ are $\F_2$-linearly independent. Thus $A\notin\C$.
Set $\D=\C_0\oplus\langle A\rangle_{\F_2}\in\addremove(\C,1,1)$. Since $A$ has rank one, $\dd(\D)=1$ by Proposition~\ref{prop:min}. Corollary~\ref{cor:blocks} with $\ell_s=1$ gives $\alpha_1=\alpha_2=3$ and $\alpha_4=\alpha_5=2$, while the $3$-dimensional MRD subcode in $\C_0$ gives $\alpha_3=3$. Hence $(\alpha_1(\D),\ldots,\alpha_6(\D))=(3,3,3,2,2,1)$. Here the index $i=3$, the only index $i<6$ at which Corollary~\ref{cor:blocks} allows a decrease, keeps the value of $\C$.

This depends on the hyperplane. Write $b_0$ for the coefficient of $1$ of $b\in\F_8$ in the basis $\gamma$, and let $\C_0'$ be the image of $\{a_0x+a_1x^2:\ a_0,a_1\in\F_8,\ (a_0)_0=(a_1)_0\}$. An exhaustive computer search shows that $\C_0'$ contains no $3$-dimensional subcode of minimum distance $3$, and that $\alpha_3(\C_0'\oplus\langle A'\rangle_{\F_2})=2$ for every $A'\in\F_2^{3\times3}\setminus\C$. For $A'=E_{11}$, $\D'=\C_0'\oplus\langle E_{11}\rangle_{\F_2}\in\addremove(\C,1,1)$ has $(\alpha_1(\D'),\ldots,\alpha_6(\D'))=(3,3,2,2,2,1)$, so the smaller value in case 2) of Corollary~\ref{cor:blocks} is attained.
\end{example}

\section{Conclusion}
We studied how subcode distances change under the Add-and-Remove construction. For a matrix Gabidulin code $\C$ and $\ell_s<m$, the subcode distances $\alpha_i(\D)$ with $i\le\dim_{\Fq}(\C_s)$ coincide with those of $\C$ except at the last $\ell_s$ indices of each block of length $m$, where they can decrease by one, and Example~\ref{ex:f8} shows that this can happen. For $\miranda$, the lower bound on the minimum distance of the dual public code used in the distinguisher of~\cite{CouvreurDGV26} (described there for $\ell_s=0$) is a special case of our bounds; for $\ell_s>0$, the same bound holds on a subcode of codimension $\ell_s$. Using the other subcode distances in a distinguisher is left for future work.

\section*{Declaration on the use of AI}
During the preparation of this manuscript, the authors used GPT-5.6 Luna to refine the text, verify the calculations, and find relevant literature. The authors explicitly declare that no mathematical ideas were generated by AI, and the authors independently verified all the theorem statements and proofs. The authors take full responsibility of the final content.

\bibliographystyle{abbrv}
\bibliography{references}

\end{document}